\documentclass[10 pt]{amsart}
\usepackage{graphicx}

\usepackage{amsthm,amsmath,latexsym,amssymb,amsfonts,amscd,hyperref}

\usepackage[all,cmtip]{xy}

\newtheorem{theorem}{Theorem}[section]
\newtheorem{lemma}[theorem]{Lemma}
\newtheorem{proposition}[theorem]{Proposition}

\theoremstyle{definition}

\newtheorem{example}[theorem]{Example}

\theoremstyle{remark}
\newtheorem{remark}[theorem]{Remark}

\numberwithin{equation}{section}

\begin{document}
\vspace{2cm}
\title[On deformations of $A_\infty$-algebras]{On deformations of $A_\infty$-algebras 
}

\author{Alexey A. Sharapov}
\address{Physics Faculty, Tomsk State University, Lenin ave. 36, Tomsk 634050, Russia}
\email{sharapov@phys.tsu.ru}
\thanks{The first author was supported in part by RFBR
Grant No. 16-02-00284 A and by Grant No. 8.1.07.2018 from ``The Tomsk State University competitiveness improvement programme''. 
The second author was partially supported by RSF Grant No. 18-72-10123 in association with Lebedev Physical Institute. 
}

\author{Evgeny D. Skvortsov}
\address{Albert Einstein Institute, 
Am M\"{u}hlenberg 1, D-14476, Potsdam-Golm, Germany}
\address{Lebedev Institute of Physics, 
Leninsky ave. 53, 119991 Moscow, Russia}
\email{evgeny.skvortsov@aei.mpg.de}




\begin{abstract}
A simple method is proposed for deforming $A_\infty$-algebras by means of the resolution technique. The method is then applied to the associative algebras of polynomial functions on quantum superspaces. Specifically, by introducing suitable resolutions, we construct explicit deformations of these algebras in the category of minimal $A_\infty$-algebras. The relation of these deformations to higher spin gravities is briefly discussed.
\end{abstract}

\maketitle

\section{Introduction}

One of the main concerns of modern algebra is the weakening various algebraic structures in a coherent way. 
Thus for each type of `classical' algebras, like associative, Lie, Poisson etc., one can associate its strong homotopy analog called, respectively, $A_\infty$ -, $L_\infty$-, $G_\infty$-, $\cdots$-algebras.   While  the classical algebras are defined in terms of binary multiplication operations obeying certain relations, the strong homotopy algebras involve the whole family of $n$-ary operations  $\{m_n\}$, from $1$ to $\infty$, hence the name\footnote{There are also versions involving  $m_0$, they are called non-flat.}. The operations are subject to infinite sets of defining relations in such a way that the binary maps $m_2$ satisfy the `classical' relations up to homotopy determined by $m_1$ and $m_3$.  As usual, the need for such a generalization of classical algebraic structures stems from various problems of physics and mathematics \cite{MSS}. 

In physics, for instance, strong homotopy algebras typically control the structure of classical equations of motion. This is best illustrated by an example of string field theory \cite{W}, \cite{Z}, \cite{GZ}, where the classical dynamics are governed by the generalized Maurer--Cartan (MC) equation
\begin{equation}\label{GMCEq}
m_1(\Phi)+m_2(\Phi,\Phi) +m_3(\Phi,\Phi,\Phi)+\cdots =0\,.
\end{equation}
Here $\Phi$ is the  string field taking values in the corresponding algebra. The type of the algebra constituted by $m$'s depends on string's topology: it is the $A_\infty$ for open strings and $L_\infty$ for closed \cite{K}, \cite{KS}, \cite{EKS}. In this particular situation the role of the differential  $m_1$ is played by the BRST operator, while the higher structure maps correspond to the tree-level string amplitudes.  The defining conditions of strong homotopy algebra reincarnate then in the form of the gauge symmetry transformations 
$$
\delta_{\Lambda}\Phi=m_1(\Lambda)+m_2(\Phi, \Lambda)+\cdots\,,
$$
$\Lambda$ being an infinitesimal gauge parameter. Given the relation between the low energy string and field theories it is little wonder that the $L_\infty$- and $A_\infty$-algebras show up in the structure of conventional field theories as well \cite{HZ}, \cite{MS}, \cite{RZ}.

Our interest to the strong homotopy algebras is mostly inspired by applications to higher spin gravities. Like string field theory, higher spin theories involve infinite collections of fields of all spins, whose interaction is governed by higher spin symmetries. At the  level of formal consistency the problem of introducing interactions \cite{Vasiliev:1988sa} (see \cite{Vasiliev:1999ba}, \cite{Didenko:2014dwa}, \cite{Sharapov:2017yde} for a review) is known to be equivalent to constructing an appropriate $L_\infty$- or $A_\infty$-algebra. We call such theories formal higher spin gravities. The first structure map $m_1$ is then given by the de Rham differential $d$ on exterior forms and the truncated system of maps $\{m_n\}_{n=2}^\infty$ defines a strong homotopy algebra by itself. The homotopy algebras of the form $\{m_n\}_{n=2}^\infty$ are called minimal.  These constitute an important category in the world of homotopy algebras as it is known that each $L_\infty$- or $A_\infty$-algebra is quasi-isomorphic to a minimal one \cite{Konts}, \cite{Kad}.  For minimal algebras  the second structure map $m_2$ satisfies the `classical' relations, so that the equations of motion for massless higher spin fields admit a consistent truncation 
\begin{equation}\label{VEq}
d\Phi=-m_2(\Phi,\Phi)\,.
\end{equation}
This is nothing but the usual MC equation associated to a differential graded algebra, the higher spin algebra.  From the physical viewpoint, Eq. (\ref{VEq}) describes the dynamics of free higher spin fields, even though the right hand side contains the fields non-linearly.  The genuine interaction vertices come from the higher structure maps $m_n$, $n\geq 3$. By construction, these deform the free gauge symmetry of equations (\ref{VEq}) in a consistent way, so that the full nonlinear system possesses the same number of physical degrees of freedom.\footnote{It is important to stress that $L_\infty$- or $A_\infty$-algebras solve only the problem of formal consistency  and further (physical) restrictions are necessary in order to have well-defined equations. Still the existence of higher spin gravities is well justified  on the basis of CFT dual descriptions, see e.g. \cite{Bekaert:2015tva}.}  Thus, given the form of free field equations (\ref{VEq}), the problem of switching on formally consistent interactions appears to be equivalent to the deformation of the underlying higher spin algebra in the category of minimal algebras.   

At an abstract  level the deformation problem for strong homotopy algebras was discussed in \cite{PSh}, \cite{FP}. 
Here we are concerned with developing practical methods for constructing deformations of $A_\infty$-algebras, particularly minimal deformations of graded associative  algebras. 
Let us outline our approach to the problem.
Given an $A_\infty$-algebra,  we first construct its resolution in the category of double $A_\infty$-algebras.  One can regarded the latter as a natural  extension of the category of double complexes. 
The choice of  a resolution is highly ambiguous.  This ambiguity, however,  provides some flexibility when dealing with particular algebras. As with double complexes, we can then define the total $A_\infty$-structure, which, by construction, is quasi-isomorphic to the original one.  If the resolution is `good enough',  the total $A_\infty$-structure admits a plenty of linear deformations, i.e., formal deformations that terminate at the first order. In many interesting cases such deformations are easy to construct and classify, especially, if one restricts to the class of non-flat deformations. The use of non-flat linear deformations is the key point of our approach. The desired deformation of the original $A_\infty$-algebra is then induced by a linear deformation of the total $A_\infty$-structure. To make this last step we apply a sort of homotopy transfer technique, which is detailed in Sec. 4 and 5. 

In the last Sec. 6,  we illustrate the above approach by constructing minimal deformations of polynomial algebras on quantum superspaces. This class of algebras, being of some interest on its own,  is closely related to  higher spin algebras in various dimensions. In particular, we compute the cohomology relevant to the minimal deformations of these algebras  and make comments on  physical interpretation of some other cocycles.

\section{$A_\infty$-algebras and their deformations}
\label{sec-1}

Throughout this paper, $k$ is an arbitrary field of characteristic zero and all unadorned tensor products $\otimes$ and Homs are taken over $k$. We start with reminding some basic definitions and constructions concerning  $A_\infty$-algebras. 

Let $V=\bigoplus V^l$ be a $\mathbb{Z}$-graded vector space and let $T(V)=\bigoplus_{n\geq 0} V^{\otimes n}$ denote its  tensor algebra with the convention that $T^0(V)=k$.  The spaces $T(V)$ and   $\mathrm{Hom}(T(V),V)$ naturally inherit the grading of $V$. Furthermore,  the $\mathbb{Z}$-graded vector  space  $\mathrm{Hom}(T(V),V)=\bigoplus \mathrm{Hom}^l(T(V),V)$ is known to carry the structure of a graded Lie algebra with respect to the Gerstenhaber bracket \cite{Gerst}. This is defined as follows. Given a pair of homogeneous homomorphisms  $f\in \mathrm{Hom}(T^n(V),V)$ and $g\in \mathrm{Hom} (T^{m}(V),V)$, we set 
\begin{equation}\label{GB}
[f,g]=f\circ g-(-1)^{|f||g|}g\circ f\,,
\end{equation}
where 
$$
(f\circ g)(v_1\otimes v_2\otimes\cdots\otimes v_{m+n-1} )$$ 
$$ =\sum_{i=0}^{n-1}(-1)^{|g|\sum_{j=1}^i|v_j|} f(v_1\otimes \cdots\otimes v_i\otimes g(v_{i+1}\otimes \cdots\otimes v_{i+m})\otimes \cdots \otimes v_{m+n-1})
$$
and $|g|$ stands for the degree of $g$ as a linear map between graded vector spaces\footnote{Here we follow \cite{KS} in defining the degree of multi-linear maps. 
The conventional $\mathbb{Z}$-grading \cite{Gerst} on $\mathrm{Hom}(T(V),V)$ is related to ours by {\it suspension}:  $V\rightarrow V[-1]$, $V[-1]^l:=V^{l-1}$. Of course, this  results in additional minus signs in the definition of the Gerstenhaber bracket. }. 

The Gerstenhaber bracket is graded skew-symmetric, 
$$
[f,g]=-(-1)^{|f||g|}[g,f]\,,
$$
and satisfies the graded Jacobi identity
$$
[[f,g],h]=[f,[g,h]]-(-1)^{|f||g|}[g,[f,h]]\,.
$$
We denote this graded Lie algebra by $L$.

By definition, an {\it $A_{\infty}$-structure} on $V$ is given by an element $m\in \mathrm{Hom}^1(T(V),V)$ of degree $1$ satisfying the MC equation 
\begin{equation}\label{MC}
[m,m]=0\,.
\end{equation}
The pair $(V, m)$ is called the {\it $A_\infty$-algebra}. Expanding the MC element $m$ into the sum $m=m_0+m_1+m_2+\cdots$ of homogeneous multi-linear maps $m_n\in \mathrm{Hom}(T^n(V),V)$ and substituting it back into (\ref{MC}), we get an infinite sequence of homogeneous  relations on $m$'s, known as Stasheff's identities \cite{St}. An $A_{\infty}$-algebra is called {\it flat} if $m_0=0$. (By definition, the zero structure map $m_0=m_0(1)$ is just an element of $V^1$.)  In the flat case, the first structure map $m_1: V^{l}\rightarrow V^{l+1}$ squares to zero, $[m_1,m_1]=2m_1^2=0$, making $V$ into a complex of vector spaces.  A flat  $A_\infty$-algebra is called {\it minimal} if $m_1=0$. For minimal algebras the second structure map  $m_2: V\otimes V\rightarrow V$ makes the space $V[-1]$ into a graded associative algebra with respect to the product 
\begin{equation}\label{bbb}
u\bullet v=(-1)^{|u|}m_2(u\otimes v)\,.
\end{equation}
The associativity condition is encoded by the Stasheff identity $[m_2,m_2]=0$. From this perspective, a graded associative algebra is just an $A_\infty$-algebra with $m=m_2$.  More generally, an $A_\infty$-algebra with $m=m_1+m_2$ is equivalent to a  differential graded algebra $(V[-1],\bullet, d)$ with the product (\ref{bbb}) and the differential $d=m_1$. Again, the Leibniz rule 
$$
d(u\bullet v)=du\bullet v+(-1)^{|u|-1}u\bullet dv
$$
is tantamount  to the Stasheff identity $[m_1,m_2]=0$. 

In this paper, we are interested in formal deformations of $A_\infty$-algebras. Let $$\mathcal{L}=L\hat{\otimes}k[[t]]=\prod_{n=0}^\infty\mathrm{Hom}(T^n(V),V)\otimes k[[t]]$$ denote the completed tensor product of $L$ and $k[[t]]$, with $t$ being a formal deformation parameter. The $\mathbb{Z}$-grading and the Gerstenhaber bracket on 
$L$ extend naturally to the space $\mathcal{L}=\bigoplus \mathcal{L}^n$ making the latter into a graded Lie algebra over $k[[t]]$. By definition,
$$
[f\otimes \alpha, g\otimes \beta]=[f,g]\otimes \alpha\beta\,,\qquad \alpha (g\otimes \beta)=g\otimes \alpha\beta
$$
for all $f,g\in {L}$ and $\alpha,\beta\in k[[t]]$.  The natural augmentation $\epsilon: k[[t]]\rightarrow k$ induces the $k$-homomorphism  $\pi: \mathcal{L}\rightarrow {L}$ of graded Lie algebras that sends the deformation parameter to zero. The MC elements of the algebra  $\mathcal{L}$ are naturally identified with  $A_\infty$-structures on the $k[[t]]$-vector space $\mathcal{V}=V\otimes k[[t]]$. By definition,
$$
(f\otimes\alpha_0)(v_1\otimes \alpha_1, \ldots, v_n\otimes \alpha_n)=f(v_1, \ldots,v_n)\otimes \alpha_0\alpha_1\cdots\alpha_n
$$
for all $f\in \mathrm{Hom}(T^n(V),V)$, $v_i\in V$ and $\alpha_j\in k[[t]]$.  

We say that an $A_\infty$-structure $m^t$ on $\mathcal{V}$ is a {\it deformation} of an $A_\infty$-structure $m$ on $V$ if $\pi(m^t)=m$. In other words, the deformed $A_\infty$-structure has the form 
\begin{equation}\label{mt}
m^t=m+m^{{}_{(1)}}t+m^{{}_{(2)}}t^2+\cdots\,, \qquad m^{{}_{(l)}}\in {L}\,. 
\end{equation}
Two deformations $m^t$ and $\tilde{m}^t$ of one and the same $A_\infty$-structure $m$ are considered as {\it equivalent} if there exists an element $f\in \mathcal{L}^0$ such that $$
\tilde{m}^t=e^{tf}m^t e^{-tf}=\sum_{n=0}^\infty\frac{t^n}{n!}(\mathrm{ad}_f)^n (m^t)\,.
$$
We denote the space of all nonequivalent deformations of $m$ by $\mathcal{M}(\mathcal{V},m)$.

The central problem of algebraic deformation theory is the construction and classification of  all the deformations of a given algebraic structure up to equivalence. The modern approach to the problem can be summarized by the following thesis: In characteristic zero, each deformation problem is governed by a differential graded Lie algebra (DGLA) \cite{GM}.   
In the case under consideration, an appropriate DGLA  can be described as follows.  Given an MC element $m\in {L}\subset \mathcal{L}$, we can regard $\mathcal{L}$ as a DGLA with the differential $\partial: \mathcal{L}^n\rightarrow \mathcal{L}^{n+1}$ given by the adjoint action of $m$, i.e., 
$$
\partial f=[m,f]\,.
$$
The DGLA $(\mathcal{L},\partial)$ contains the differential ideal $\bar{\mathcal{L}}=t\mathcal{L}=\ker\,\pi$.  
Define the space of MC elements of $(\bar{\mathcal{L}}, \partial)$ as
$$
\mathcal{MC}(\bar{\mathcal{L}}, \partial)=\Big\{\mu\in \bar{\mathcal{L}}^1 \;\Big |\; \partial\mu+\frac12[\mu,\mu]=0\,\Big\}\,.
$$
Notice that the elements of $\bar{\mathcal{L}}^0$ constitute a Lie subalgebra, whose exponentiation gives the formal group $\Gamma_{\bar{\mathcal{L}}} =\exp \bar{\mathcal{L}}^0$, sometimes called the gauge group.  The adjoint action of $\Gamma_{\bar{\mathcal{L}}}$ on $\bar{\mathcal{L}}^1$ induces  transformations of the space $\mathcal{MC}(\bar{\mathcal{L}},\partial)$, namely,
$$
\mu \mapsto \mu'= e^\gamma(m+\mu)e^{-\gamma} -m=-m +\sum_{n=0}^\infty \frac{1}{n!}\mathrm{ad}_\gamma^n(m+\mu)
$$
for all $\gamma \in \bar{\mathcal{L}}^0$. The {\it formal moduli space} of the DGLA $(\bar{\mathcal{L}}, \partial)$ is defined now to be the quotient 
$$
\mathcal{M}(\bar{\mathcal{L}},\partial)=\mathcal{MC}(\bar{\mathcal{L}}, \partial)/\Gamma_{\bar{\mathcal{L}}}\,.
$$
Writing now the deformed $A_\infty$-structure (\ref{mt}) as $m^t=m+\mu$, one can see that the points of the formal moduli space are in one-to-one correspondence with the equivalence classes of formal deformations of the $A_\infty$-algebra $(V, m)$.  This allows one to identify the two spaces, $$\mathcal{M}(\bar{\mathcal{L}},\partial)=\mathcal{M}(\mathcal{V}, m)\,.$$   

It is clear that every morphism $\alpha: \bar{\mathcal{L}}\rightarrow \bar{\mathcal{N}} $ of two DGLAs governing the deformation problem above induces a morphism  $\hat{\alpha}: \mathcal{M}(\bar{\mathcal{L}}, \partial)\rightarrow \mathcal{M}(\bar{\mathcal{N}},\partial)$ of the corresponding moduli spaces as well as a homomorphism $H(\alpha): H(\bar{\mathcal{L}},\partial)\rightarrow H(\bar{\mathcal{N}},\partial)$  of the cohomology groups. By definition, the morphism $\alpha$ is called {\it quasi-isomorphism} if it induces an isomorphism in cohomology.  The next statement may be regarded as a basic theorem of deformation theory
\begin{theorem}\label{Th1}
Let $\alpha: \bar{\mathcal{L}}\rightarrow \bar{\mathcal{N}}$ be a quasi-isomorphism of DGLAs. Then the induced map $\hat{\alpha}: \mathcal{M}(\bar{\mathcal{L}}, \partial)\rightarrow \mathcal{M}(\bar{\mathcal{N}},\partial)$ is an isomorphism.
\end{theorem}

In other words, quasi-isomorphic DGLAs  give rise  to equivalent deformation problems. The proof can be found in \cite{Konts}, \cite{GM}. 

\section{Resolution of an $A_\infty$-algebra}\label{co-res}
We say that two $A_\infty$-structures $m'$ and  $m''$ on the same vector space  $V$ are {\it compatible},  if 
\begin{equation}\label{C-C}
    [m',m'']=0\,.
\end{equation}
In this case, the sum $m=m'+m''$ is again an   $A_\infty$-structure (actually, any linear combination of the two gives an $A_\infty$-structure).  

Suppose now that the $\mathbb{Z}$-grading on $V$ comes from a bi-grading, that is, $V=\bigoplus V^{p,q}$ and $V^l=\bigoplus_{p+q=l}V^{p,q}$. We will denote the bi-degree of a homogeneous element 
$a\in V$ by
\begin{equation}
    \deg a=(p,q)
\end{equation}
and refer to $|a|=p+q$ as the {\it total degree} of $a$. The double gradation of $V$ allows us to consider $A_\infty$-structures that are homogeneous with respect to the first and second degrees.

Given a bi-graded  vector space $V=\bigoplus V^{p,q}$, a {\it double $A_\infty$-structure} on $V$ is given by a pair of compatible $A_\infty$-structures $m'$ and $m''$ of bi-degrees
\begin{equation}
    \deg m'=(1,0)\,,\qquad \deg m''=(0,1)\,.
\end{equation}

In the special  case that $m'=m'_1$ and $m''=m''_1$ the double $A_\infty$-algebra $(V, m',m'')$ degenerates to a double complex of vector spaces.  In the following we will mostly interested in the case where the second $A_\infty$-structure is given simply by a differential $d=m''_1$, while the first one is arbitrary. Then the  compatibility condition (\ref{C-C}) takes  the form
$$
    dm'_k(v_1\otimes v_2\otimes \cdots\otimes v_k)=-\sum_{l=1}^k(-1)^{|v_1|+\cdots +|v_{l-1}|}m'_k(v_1\otimes \cdots \otimes dv_l\otimes \cdots \otimes v_k)
$$
for all $k=0,1,\ldots$, and $v_i\in V$. 

By construction, the complex $(V,d)$ splits into the direct sum of subcomplexes $(V^{p,\bullet}, d)$ labeled by the first degree. Let us further assume that $V^{p,q}=0$ for all $q<0$, and that the differential $d: V^{p,q}\rightarrow V^{p,q+1}$ is acyclic in positive $q$ degrees, that is, $H^\bullet(V,d)\simeq H^0(V,d)=\ker d\cap V^{\bullet,0}$. Due to the compatibility and degree conditions the first $A_\infty$-structure  $m'$ on $V$ can be consistently restricted onto the subspace $W=H^0(V,d)$ making the latter into an $A_\infty$-algebra. Let us denote this restriction by ${m}^{}_W=m'|_W$. In this situation we say that the triple $(V, m', d)$ is a {\it resolution} of the $A_\infty$-algebra $(W,{m}_W)$ and refer to the second degree as the {\it resolution degree}.  

The main idea behind our approach is to deform the $A_\infty$-algebra $(W,m_W)$ by deforming its suitable resolution. The construction goes as follows. 

Summing up the compatible $A_\infty$-structures, we endow $V$ with the `total' $A_\infty$-structure $m=m'+d$. Then, following the general philosophy discussed in the previous section, we introduce the DGLA $\mathcal{L}=\bigoplus\mathcal{L}^n$, where 
$$
\mathcal{L}^n=\bigoplus_{\stackrel{{p+q=n}}{ {}_{q\geq 0}}}\mathcal{L}^{p,q}\,,\qquad \mathcal{L}^{p,q}=\mathrm{Hom}^{p,q}(T(V),V)\hat{\otimes}k[[t]]\,.
$$
Notice that we include into $\mathcal{L}$ only homomorphisms of non-negative resolution degree. The Lie bracket in $\mathcal{L}$ is given by the Gerstenhaber bracket, while the adjoint action of $m$ endows  $\mathcal{L}$ with the differential
$$
\partial =[m, - ]\,,\qquad \partial :\mathcal{L}^n\rightarrow \mathcal{L}^{n+1}\,.
$$
Recall that $\bar{\mathcal{L}}=t\mathcal{L}$ denotes the differential ideal of $\mathcal{L}$ that governs the formal deformations of the $A_\infty$-algebra $(V, m)$.

As the next step, we evaluate the cohomology of the DGLA $(\bar{\mathcal{L}}, \partial)$. To this end, we split the  differential into the sum $\partial =\partial'+\partial''$ of the vertical and horizontal differentials 
$$
\partial'=[m', -]\,,\qquad \partial'' =[d, - ]
$$
and, using the bicomplex structure, endow $\bar{\mathcal{L}}$ with a decreasing filtration associated to the first degree:
$$
F^p\bar{\mathcal{L}}=\bigoplus_{s\geq p} \bar{\mathcal{L}}^{s, \bullet}\,,\qquad F^p\bar{\mathcal{L}} \supset F^{p+1}\bar{\mathcal{L}}\,, \qquad \bigcup_{p\in \mathbb{Z}} F^p\bar{\mathcal{L}}=\bar{\mathcal{L}}\,.
$$
Since the resolution degree is bounded below, the filtration is regular and
yields a spectral sequence $\{E^{\bullet,\bullet}_r\}$ with $E_2^{p,q}=H^p_{\partial'} H^q_{\partial''}(\bar{\mathcal{L}})$.  We claim that 

\begin{proposition}\label{prop1}
The cohomology of the complex $(\bar{\mathcal{L}}, \partial'')$ is centered in resolution degree zero, so that $H^q_{\partial''}(\bar{\mathcal{L}})=0$ for $q>0$. 
\end{proposition}

\begin{proof}
To compute the groups $H^q_{\partial''}(\bar{\mathcal{L}})$, we split the complex $(V,d)$ by introducing a contracting homotopy $h: V^{p,q}\rightarrow V^{p,q-1}$ together with the inclusion  $i:W \rightarrow V$ and the projection $p: V\rightarrow W$ mappings associated to the subspace $W=H^0(V,d)\subset V$.  Without loss in generality we may assume that 
$$
hd+dh=1_V-ip\,, \qquad h^2=0\,,\qquad pi=1_W\,,\qquad ph=0\,,\qquad  hi=0\,.
$$
Then the operator $h$ induces a contracting homotopy on $\tilde {h}: \bar{\mathcal{L}}^{\bullet,q}\rightarrow \bar{\mathcal{L}}^{\bullet,q-1}$ defined by 
\begin{equation}\label{hat-h}
(\tilde{h} f)(v_1\otimes\cdots\otimes v_n)=hf(v_1\otimes\cdots\otimes v_n)\,.
\end{equation}
It is clear that
$$
\tilde{h}\partial+\partial \tilde{h}=1_{\bar{\mathcal{L}}}-\tilde{i}\tilde{p}\,,\qquad \tilde{h}^2=0\,,
$$
where the operators $\tilde{i}$ and $\tilde{p}$ are defined similar to (\ref{hat-h}). Since $\ker(1-ip)=W$,
this means that any nontrivial cocycle (homomorphisms) of $(\bar{\mathcal{L}},\partial'')$ is cohomologous to one taking values in the subspace $W$. On the other hand, any homomorphism of $\bar{\mathcal{L}}$ with values in $W$ has resolution degree zero and is automatically a nontrivial $\partial''$-cocycle. Thus,
\begin{equation}\label{prop}
H^\bullet_{\partial''}(\bar{\mathcal{L}})\simeq  H^0_{\partial''}(\bar{\mathcal{L}})= \ker \partial'' \cap \bar{\mathcal{L}}^{\bullet,0} \,.
\end{equation}
\end{proof}
We see that all the nonzero groups $E_1^{p,q}$,  and hence $E^{p,q}_2$, are nested on the base ($q=0$). As a result the spectral sequence collapses at the second term  giving  the isomorphism 
\begin{equation}\label{H-E}
H^n(\bar{\mathcal{L}}, \partial)\simeq E_2^{n,0}=H^n_{\partial'} H^0_{\partial''}(\bar{\mathcal{L}})\,. 
\end{equation}

Rel. (\ref{prop}) identifies the group $H^0_{\partial''}(\bar{\mathcal{L}})$ with the subspace $\bar{\mathcal{N}}=\ker \partial'' \cap \bar{\mathcal{L}}^{\bullet,0}\subset \bar{\mathcal{L}}$.  Actually, $\bar{\mathcal{N}}$ is not just a subspace but a differential subalgebra of $(\bar{\mathcal{L}}, \partial)$ as one can easily see. Notice also that $m'\in \mathcal{N}=\ker \partial\cap \mathcal{L}^{\bullet,0}$ and $\partial|_{\bar{\mathcal{N}}}=\partial'$. This allows us to interpret Rel. (\ref{H-E}) as an isomorphism of the cohomology groups:
$$
H^\bullet (\bar{\mathcal{L}}, \partial)\simeq H^{\bullet}(\bar{\mathcal{N}},\partial')\,.
$$
In other words, the natural inclusion $\alpha:\bar{\mathcal{N}}\rightarrow \bar{\mathcal{L}}$ is a quasi-isomorphism of DGLAs. According to Theorem \ref{Th1}, this implies an isomorphism of the formal moduli spaces $$\mathcal{M}(\bar{\mathcal{L}}, \partial)\simeq \mathcal{M}(\bar{\mathcal{N}},\partial')\,.$$ 

It is well known that each quasi-isomorphism of two DGLAs has a {\it quasi-inverse} homomorphism, see e.g. \cite[Sec. 4.1]{DHR}. 
Therefore, there exists a homomorphism 
\begin{equation}\label{be}
\beta :\bar{\mathcal{L}}\rightarrow \bar{\mathcal{N}}
\end{equation}
such that the induced homomorphism $H(\beta\alpha)=H(\beta)H(\alpha)$ defines the identical mapping on $H^\bullet (\bar{\mathcal{N}},\partial')$. With the help of $\beta$ we can transfer the MC elements backwards: If $\gamma \in \mathcal{MC}(\bar{\mathcal{L}}, \partial)$, then $\mu=\beta(\gamma)\in \mathcal{MC}(\bar{\mathcal{N}},\partial')$. 

It remains to note that due to the condition
$$
[d, \mu]=0
$$
each MC element $\mu\in \mathcal{MC}(\bar{\mathcal{N}},\partial')$ admits a consistent restriction to the subspace $W\subset V$, in the sense that $T(W)\subset T(V)$ and $\mu: T(W)\rightarrow W$. Let us denote this restriction by ${\mu}|_W$.  Combining the quasi-isomorphism (\ref{be}) with the restriction map, we can deform the original $A_\infty$-structure $m^{}_W=m'|^{}_{W}$ on $W$ by the formula
\begin{equation}\label{MF}
m_W^t=m^{}_W+{\beta(\gamma)}|^{}_W \qquad \forall \gamma \in \mathcal{MC}(\bar{\mathcal{L}},\partial)\,.
\end{equation}
The problem now is to find out an explicit formula for the quasi-isomorphism (\ref{be}). This will be discussed in the next two sections.

\section{Transferring $A_\infty$-structures}

As we have seen any deformation of the $A_\infty$-algebra $(V, m)$ induces a deformation of the algebra $(V, m')$ and -- through the restriction -- a deformation of $(W, m_W)$. Although the deformation problems for the $A_\infty$-structures $m$ and $m'$ are essentially equivalent, the former provides more freedom for constructing formal deformations as we are not restricted to the zero resolution degree.  

In this paper, we focus  upon a special class of deformations of  $(V,m)$ that are represented by straight lines in the MC space $\mathcal{MC}(\bar{\mathcal{L}}, \partial)$. Any such deformation is defined by an $A_\infty$-structure  $\lambda$
 which is compatible with $m$, i.e.,
\begin{equation}\label{comp-con}
[\lambda,\lambda]=0\,,\qquad [m,\lambda]=0\,.
\end{equation}
This ensures that the formal line
\begin{equation}\label{m_t}
m^t=m+t\lambda 
\end{equation}
defines a family of $A_\infty$-structures on $V$. Formula (\ref{MF}) yields then a formal deformation of the $A_\infty$-structure on $W$:
\begin{equation}\label{m_W}
m_W^t=m^{}_W+\beta(t\lambda)|^{}_W\,.
\end{equation}
In general, the resulting $A_\infty$-structure $m_W^t$ may contain higher orders in $t$, defining  a formal curve rather than a line in the MC space.

A simple observation concerning the linear deformations (\ref{m_t}) is that we can always satisfy the quadratic relation  (\ref{comp-con}) by choosing  $\lambda \in \mathrm{Hom}(T^0(V), V)$. Having no arguments, the `homomorphism' $\lambda$ automatically satisfies the first equation in (\ref{comp-con}) and we are left with the only linear condition. The latter can easily be analyzed  in many practical cases. For example, let $A=(V, m')$ be a graded associative algebra with $m'=m'_2$ and let $\lambda\in V$. Then the second equation in  (\ref{comp-con}) tells  us that 
$$
d\lambda=0\,, \qquad [m_2',\lambda]=0\,.
$$
In other words,  $\lambda$ is just a $d$-cocycle belonging to the center of the associative algebra $A$.  
By construction, any such cocycle gives rise to a deformation of the associative algebra $A_W=(W, m_W)$ in the category of $A_\infty$-algebras. It is worth noting that the deformed $A_\infty$-structure (\ref{m_W})  may well be flat, while its preimage (\ref{m_t}) is not.

In order to construct the desired MC element $\mu=\beta(t\lambda)\in \bar{\mathcal{L}}^{1,0}$ we follow the method of our recent work  \cite{SkSh}. Namely, we introduce the operators 
$$
\partial_\mu=\partial +[\mu, - ]\,,\qquad N=t\frac{d}{dt}
$$
together with the pair of auxiliary elements $\Gamma\in \bar{\mathcal{L}}^0$ and $\Lambda\in \mathcal{L}^1$ of total degrees $0$ and $1$. The unknowns  $\Gamma$, $\Lambda$ and $\mu$ are supposed to satisfy the following set of `master equations':
\begin{equation}\label{MEq}
     \partial_\mu\Gamma=t\Lambda-N\mu\,,\qquad N\Lambda=[\Gamma,\Lambda]\,,\qquad \pi(\Lambda)=\lambda\,.
\end{equation}
The name and the relevance of these equations to our problem are explained by the next statement. 

\begin{lemma}
The element $\mu\in \bar{\mathcal{L}}^{1,0}$  defined by Eqs. (\ref{MEq}) satisfies the MC equation 
$$
\partial \mu=-\frac12[\mu,\mu]
$$
whenever 
$$
\partial \lambda =0\,,\qquad [\lambda,\lambda]=0\,.
$$
\end{lemma}

\begin{proof}
Let us denote 
$$
R=\partial\mu+\frac12[\mu,\mu]\,,\qquad T=\partial_\mu \Lambda\,,\qquad S=[\Lambda,\Lambda]\,.
$$
Applying the operator $\partial_\mu$ to both sides of the master equations (\ref{MEq}), we find  
\begin{equation}\label{RT}
[R,\Gamma]=tT -NR \,,\qquad NT=[\Gamma, T]+tS\,, 
\end{equation}
provided that $\Lambda$, $\Gamma$, and $\mu$ obey (\ref{MEq}). Acting by $N$ on $S$ and using the master equations once again, we get one more relation 
\begin{equation}\label{S}
NS=[\Gamma,S]\,.
\end{equation}
Taken together, Eqs. (\ref{RT}) and (\ref{S}) constitute a closed system of linear ODEs
$$
\dot{R}=T-\frac1t[R,\Gamma]\,,\qquad \dot{T}=S+\frac1t[\Gamma,T]\,,\qquad \dot S=\frac1t[\Gamma,S]\,.
$$
Since $\Gamma \in \bar{\mathcal{L}}$, the right hand sides of these equations are regular in $t$. Therefore, the equations have a unique solution $R=0$, $T=0$, and $S=0$ subject to the  initial conditions
 $$
 R(0)=0\,,\qquad T(0)=\partial\lambda=0\,,\qquad S(0)=[\lambda,\lambda]=0\,.
 $$
\end{proof}

It remains to show that the master equations (\ref{MEq}) do have a solution. 

\begin{lemma}\label{lem2}
    Eqs. (\ref{MEq}) have a unique solution satisfying the additional conditions 
    \begin{equation}\label{bc}
      \tilde{h}\Gamma=0\,,\qquad  \tilde{p}\Gamma=0\,,
    \end{equation}
where the operators $\tilde{h}$ and $\tilde{p}$ are defined by Eq. (\ref{hat-h}).
\end{lemma}
\begin{proof} Let us expand $\Gamma$ and $\Lambda$ in homogeneous components:
\begin{equation}\label{2exp}
\Gamma=\sum_{n=0}^\infty\Gamma_n\,,\qquad \Lambda=\sum_{n=0}^\infty\Lambda_n\,,
\end{equation}
where
$$
\deg \Gamma_n=(-n,n)\,,\qquad \deg \Lambda_n=(1-n,n)\,.
$$
On substitution of these expansions into the master equations  (\ref{MEq}), we obtain the system  of homogeneous equations 
\begin{equation}\label{Nm}
N\mu=-\partial_\mu'\Gamma_0+t\Lambda_0\,, \qquad\qquad\qquad\qquad\quad
\end{equation}
\begin{equation}\label{dG}
  \partial''\Gamma_{n}+\partial'_\mu\Gamma_{n+1}=t\Lambda_{n+1}\,,\quad n=0,1,2,\ldots,
\end{equation}
\begin{equation}\label{Nl}
    N\Lambda_n=\sum_{m=0}^n[\Gamma_{n-m},\Lambda_m]\,,\quad n=0,1,2,\ldots\,.
\end{equation}
Here we introduced the shorthand notation $\partial'_\mu=\partial'+[\mu,-]$.
Applying the contracting homotopy operator (\ref{hat-h}) to  equations (\ref{dG}) and using conditions (\ref{bc}), we can formally solve (\ref{dG}) for $\Gamma$ as  
\begin{equation}\label{g-ser}
\Gamma_n=t\sum_{k=0}^\infty (-\tilde{h}\partial'_\mu)^k\tilde{h}\Lambda_{n+k+1}\,.
\end{equation}
Substituting this expression into the remaining equations  (\ref{Nm}) and (\ref{Nl}), we get the system of  ODEs 
\begin{equation}\label{dif-eq}
\dot \mu =\sum_{m=0}^\infty (-\partial'_\mu\tilde{h})^{m}\Lambda_{m}\,,\qquad \dot \Lambda_n =\sum_{k=0}^n\sum_{m=0}^\infty [ \tilde{h}(-\partial'_\mu\tilde{h})^m\Lambda_{k+m+1},\Lambda_{n-k}]\,,
\end{equation}
where the overdot stands for the derivative in $t$.  These last equations can be solved by iterations giving a unique solution subject to the initial conditions $\mu(0)=0$ and $\Lambda(0)=\lambda$. 
In particular,  if
$$
\lambda=\sum_{m=0}^M \lambda_m
$$
is the  expansion of $\lambda$ with respect to the resolution degree, then the first-order deformation is determined 
by 
\begin{equation}\label{dm}
\dot{\mu}(0)=\sum_{m=0}^M (-\partial'\tilde{h})^{m}\lambda_{m}\,.
\end{equation}
The expression for the second-order deformation is more cumbersome. Differentiating the first equation in (\ref{dif-eq}) and setting $t=0$, we find 
$$
\ddot{\mu}(0)=-\sum_{m=0}^M\sum_{k=1}^{m-1}(-\partial'\tilde{h})^k[{\dot{\mu}(0)}, \tilde{h}(-\partial'\tilde{h})^{m-k-1}\lambda_m]
+\sum_{m=0}^M(-\partial'\tilde{h})^m\dot{\Lambda}_m(0)\,,
$$
where
$$
\dot{\Lambda}_n(0)=\sum_{k=0}^n \sum_{m=0}^{M-k-1} [\tilde{h}(-\partial'\tilde{h})^{m}\lambda_{k+m+1},\lambda_{n-k}]
$$
and $\dot{\mu}(0)$ is given by (\ref{dm}). As is seen all the sums are finite and this property holds true in higher orders. 
\end{proof}

\begin{remark}\label{rem} In the above proof, the convergence of the series (\ref{g-ser}) followed  a posteriori, after solving the differential equations.  In many interesting cases, however, it can be ensured a priori. Suppose, for example,
$$
V=\bigoplus_{\stackrel{-m\leq p\leq 0}{_{q\geq 0}}}V^{p,q}\,,
$$
that is, the first degree of homogeneous vectors is non-positive and bounded below by $-m$. Then, so is the first degree of the associated DGLA: 
$$
\mathcal{L}=\bigoplus_{\stackrel{-m\leq p\leq 0}{_{q\geq 0}}}\mathcal{L}^{p,q}\,.
$$
As a result, the expansions (\ref{2exp}) are finite and the series (\ref{g-ser}) contains only finite number of  terms. 

Notice that the space $\mathcal{L}=\bigoplus \mathcal{L}^{p,q}$ has two more natural gradings in addition to the original bi-grading. 
These are given by the degree in $t$ and by the degree of $f\in \mathcal{L}$ as an element of the graded space $\prod_{m\geq 0}\mathrm{Hom}(T^{m}(V),V)$. More precisely, $\mathcal{L}=\prod_{n,m\geq 0}\mathcal{L}_{n,m}$, where $\mathcal{L}_{n,m}$ is spanned by the elements of the form
$$
a\otimes t^n\,,\qquad a\in \mathrm{Hom}(T^m(V),V)\,.
$$
Summing up these two gradings, we obtain an $\mathbb{N}$-graded space $\mathcal{L}=\prod_{k\geq 0}\mathcal{L}_k$ with $\mathcal{L}_k=\bigoplus_{n+m=k}\mathcal{L}_{n,m}$.  Let us now suppose that 
$$
m'\in \prod_{k\geq 2}\mathrm{Hom}(T^k(V), V)\,.
$$
Then the operator $\partial'_\mu=\partial'+[\mu, - ]$ increases the $\mathbb{N}$-degree and the series (\ref{g-ser}) is well defined  as an element of $\mathcal{L}=\prod_{k\geq 0}\mathcal{L}_k$.  
\end{remark}

\section{Interpretation  via homological perturbation theory}

In the previous section, we have shown how to construct a deformation of an $A_\infty$-algebra $(V,m')$ starting form a suitable resolution $(V,m',d)$ and a compatible $A_\infty$-structure $\lambda$.  Although the  master equations (\ref{MEq}) provide an explicit solution to the deformation problem, their origin remains obscure. In order to clarify our construction,  we will put it in a slightly different approach of {\it homological perturbation theory} (HPT). A detailed account of the theory can be found in \cite{HK}, \cite{GLS}, \cite{Kr} (see also \cite{LZ} for a recent discussion of HPT in the context of formal higher spin gravities). Below we briefly review some basic definitions and statements.

First, we note that the  complex $(V,d)$, being taken together with the contracting homotopy $h$, provides  a particular example of  a {\it strong deformation retract} (SDR).  
In general, a SDR is given by a pair of complexes $(V,d_V)$ and $(W,d_W)$ together with chain maps $p:V\rightarrow W$ and $i:W\rightarrow V$
such that $pi=1_W$ and $ip$ is homotopic to $1_V$. The last property implies the  existence of a map $h: V\rightarrow V$ such that 
$$
dh+hd=ip-1_V\,.
$$
Without loss in generality, one may assume the following {\it annihilation properties}:  
$$
hi=0\,,\qquad ph=0\,,\qquad h^2=0\,.
$$
All these data can be summarized by a single  diagram 
\begin{equation}\label{SDR}
\xymatrix{
*{\hspace{5ex}(V,d_V)\;}\ar@(ul,dl)[]_{h} \ar@<0.5ex>[r]^-p
& (W, d_W) \ar@<0.5ex>[l]^-i}\,.
\end{equation}

The  situation considered in  the previous section corresponds to a special case where  $W=H(V,d_V)$ is the cohomology space of the complex $(V,d_V)$ and $d_W=0$. 

The main concern of HPT is  transferring various algebraic structures form one object to another through a homotopy equivalence. Whenever applicable, the theory provides effective algorithms and explicit formulas as distinct from the most part of classical homological algebra.  The cornerstone of HPT is 
the following statement, often called the Basic Perturbation Lemma. 
\begin{lemma}[\cite{B}]\label{BPL}
Given SDR data (\ref{SDR}) and a small perturbation $\delta$ of  $d_V$ such that $(d_V+\delta)^2=0$ and $1-\delta h$ is invertible, there is a new SDR 
$$
\xymatrix{
*{\hspace{9ex}(V,d_V+\delta )\;}\ar@(ul,dl)[]_{h'} \ar@<0.5ex>[r]^-{p'}
& (W, d_W) \ar@<0.5ex>[l]^-{i'}}\,,
$$
where the maps are given by 
$$
\begin{array}{ll}
     p'=p+p(1-\delta h)^{-1}\delta h\,,&\quad i'=i+h(1-\delta h)^{-1}\delta i\,, \\[3mm]
    h'=h+h(1-\delta h)^{-1}\delta h\,, & \quad d'_W = d_W+p(1-\delta h)^{-1}\delta i\,.
\end{array}
$$
\end{lemma}

One can think of the operator $A=(1-\delta h)^{-1}$ as being defined by a geometric series 
$$
A=\sum_{n=0}^\infty (\delta h)^n\,. 
$$
In many practical cases its convergence is ensured by the existence of a natural decreasing filtration of $V$ which is lowered by the operator  $\delta h$. 

We are concerned with transferring  $A_\infty$-structures on $V$ to its cohomology space $W$. 
To put this transference problem into the framework of HPT one first applies the tensor-space functor $T$  to the vector spaces $V$ and $W$. Recall that, in addition to the associative algebra structure, the space $T(V)$ carries the structure of a coassociative  coalgebra with respect to the Alexander--Whitney coproduct 
$$
\Delta: T(V)\rightarrow T(V)\otimes T(V)\,,
$$

$$
\Delta (v_1\otimes \cdots\otimes v_n)=1\otimes (v_1\otimes\cdots \otimes v_n)+\sum_{i=1}^{n-1}(v_1\otimes\cdots\otimes v_i)\otimes (v_{i+1}\otimes\cdots\otimes v_n)$$
$$
+(v_1\otimes\cdots \otimes v_n)\otimes 1\,.
$$
Coassociativity is expressed by the relation $(1\otimes \Delta)\Delta=(\Delta\otimes 1)\Delta$.

A linear map $D: T(V)\rightarrow T(V)$ is called a {\it coderivation}, if it obeys the co-Leibniz rule $$\Delta D=(D\otimes 1+1\otimes D)\Delta\,.$$ 
The space of coderivations is known to be  isomorphic to the space of homomorphisms $\mathrm{Hom}(T(V),V)$, so that  any homomorphism $f: T(V)\rightarrow V$ induces a coderivation $\hat{f}: T(V)\rightarrow T(V)$ and vice versa:  if $f\in \mathrm{Hom}(T^m(V),V)$, then  
\begin{equation}\label{f1}
\begin{array}{rl}
\hat{f}(v_1\otimes\cdots\otimes v_n)=\displaystyle \sum_{i=1}^{n-m+1}&(-1)^{|f|(|v_1|+\cdots+ |v_{i-1}|)}v_1\otimes\cdots\otimes v_{i-1}\\[4mm]
&\otimes f(v_i\otimes \cdots \otimes v_{i+m-1})\otimes v_{i+m}\otimes \cdots\otimes v_n
\end{array}
\end{equation}
for $n\geq m$ and zero otherwise. Among other things, this allows one to interpret the Gerstenhaber bracket (\ref{GB}) as the commutator of two coderivations.  For $f\in \mathrm{Hom}(V,V)$ the above relation reduces to the usual Leibniz rule for the tensor product.

The next statement, called the {\it tensor trick}, allows one to transfer SDR data from spaces to their tensor (co)algebras. 
\begin{lemma}[\cite{GL}]\label{TT}
With any SDR data (\ref{SDR}) we can associate a new SDR 
$$
\xymatrix{
*{\hspace{8ex}(T(V),\hat{d}_V )\;}\ar@(ul,dl)[]_-{\hat{h}} \ar@<0.5ex>[r]^-{\hat{p}}
& (T(W), \hat{d}_W) \ar@<0.5ex>[l]^-{\hat{i}}}\,,
$$ 
where the new differentials $\hat{d}_V$ and $\hat{d}_W$ are defined by the rule (\ref{f1}),
$$
\hat{p}=\sum_{n=1}^\infty p^{\otimes n}\,,\qquad \hat{i}= \sum_{n=1}^\infty i^{\otimes n}\,,
$$
and the new homotopy is given by
$$
\hat{h}=\sum_{n\geq 1}^\infty \sum_{i=0}^\infty 1^{\otimes i}\otimes h\otimes (ip)^{\otimes n-i-1}\,.
$$
\end{lemma}

After reminding the basics of HPT let us return to our deformation problem. Given a resolution $(V, m', d)$, we can define  the SDR associated to the complex $(V,d)$ and its cohomology space $(W,0)$; the mappings $p$, $i$ and $h$ are defined as in the proof of Proposition \ref{prop1}. Applying the tensor trick yields then an SDR for the corresponding tensor (co)algebras
$$
\xymatrix{
*{\hspace{7ex}(T(V),\hat{d})\;}\ar@(ul,dl)[]_-{\hat{h}} \ar@<0.5ex>[r]^-{\hat{p}}
& (T(W), 0 ) \ar@<0.5ex>[l]^-{\hat{i}}}\,,
$$
The deformed $A_\infty$-structure $m^t=m'+t\lambda$ on $V$ gives rise to the coderivation $\hat{m}^t=\hat{m}'+t\hat{\lambda}$ that squares to zero and commutes with $\hat{d}$. This allows us to treat $\hat{m}^t$ as a small perturbation of $\hat{d}$ and, by making use of the Basic Perturbation Lemma, we  arrive at the SDR 
$$
\xymatrix{
*{\hspace{13ex}(T(V),\hat{d}+\hat{m}^t)\;}\ar@(ul,dl)[]_-{\hat{h}'} \ar@<0.5ex>[r]^-{\hat{p}'}
& (T(W), \hat{m}^t_W ) \ar@<0.5ex>[l]^-{\hat{i}'}}\,.
$$
By Lemma \ref{BPL}, the differential on the right is given by 
\begin{equation}\label{hatmt}
  \hat{m}^t_W = \hat{p}(1-\hat{m}^t \hat{h})^{-1}\hat{m}^t\hat{i}\,.
\end{equation}
A suitable decreasing filtration ensuring the invertibility of  the operator $(1-\hat{m}^t \hat{h})$ comes from the total grading that combines the first degree of $V$ with the degree in $t$, see Remark \ref{rem}. We can simplify Rel. (\ref{hatmt}) by noting that $\mathrm{Im} \, \hat{m}'\hat{i}\subset W$, and hence $\hat{h}\hat{m}'\hat{i}=0$. Then the expansion of (\ref{hatmt}) in powers of $t$ takes the form 
$$
\hat{m}^t_W = \hat{m}^{}_W+ t \hat{p}(1-\hat{m}' \hat{h})^{-1}\hat{\lambda}\hat{i}+O(t^2)\,,
$$
 where $\hat{m}^{}_W=\hat{p}\hat{m}'\hat{i}$.
 In the special case that $\lambda$ is homogeneous of bi-degree $(1-r, r)$, the expression for the first-order correction can  further be simplified. Since $\hat{p}$ annihilates the elements of nonzero resolution degree, we can write 
 $$
 \hat{m}^t_W= \hat{m}^{}_W+ t \hat{p}(\hat{m}'\hat{h})^r\hat{\lambda}\hat{i}+O(t^2)\,,
 $$
 cf. (\ref{dm}). Finally, `removing the hats' of (\ref{hatmt}), we obtain the deformed $A_\infty$-structure $m^t_W$. 

\section{Examples of deformations}

In this section, we illustrate the above machinery of deformations by applying it to some bimodules over polynomial and Weyl algebras and to quantum polynomial superalgebras. Our interest to this class of examples is not purely algebraic. As indicated in Example \ref{example} below, these algebras and their deformations are of primary importance for higher spin theory.

\subsection{Minimal deformations of bimodules} Let us start with some general remarks. Given an associative algebra $A$ and  an $A$-bimodule $M$, one can define a new associative algebra $\mathcal{A}$, called the trivial extension of $A$ by the bimodule $M$. As a vector space $\mathcal{A}=A\oplus M$ and multiplication is defined by the formula  
$$
(a,m)(a',m)=(aa', am+ma')\qquad \forall a,a'\in A\,,\quad \forall m,m'\in M\,.
$$
If no extra structure is assumed, one may only deform the pair $(A,M)$ in the category of bimodules over associative algebras. This is the concern of classical deformation theory.  Notice, however, that the algebra $\mathcal{A}$ admits a natural grading. This is obtained by prescribing the spaces $A$ and $M$ the degrees $0$ and $1$, respectively.  When treated as a graded associative algebra, $\mathcal{A}=\mathcal{A}^0\oplus \mathcal{A}^1$ may have nontrivial deformations in the category of $A_\infty$-algebras. We say that the {\it deformation is minimal} if the resulting $A_\infty$-algebra is minimal.  
As a particular case, this includes the deformation problem for the original bimodule structure.   

In the following, we restrict our consideration  to a rather special yet important class of bimodules that originate from polynomial algebras endowed with automorphisms. 
Let $V$ be an $n$-dimensional vector space over $k$ and let $\vartheta: V\rightarrow V$ be an automorphism of $V$. 
The action of $\vartheta$ on $V$ induces an automorphism of the dual space $V^\ast$, which then extends to an automorphism of the symmetric algebra $A=S(V^\ast)$. Let ${}^\vartheta\! a$ denote the result of the action of $\vartheta$ on $a\in A$.   Given the automorphism $\vartheta$, we can view  the $k$-vector space $A$ as an $A$-bimodule with respect to the following left and right actions:
$$
a\circ m=am,\qquad m\circ a=m^{\vartheta}\!a\qquad \forall a,m\in A\,.
$$
As is seen the right action of $A$ on itself is twisted by  $\vartheta$. We denote this $A$-bimodule by $A^\vartheta$.  In a similar way one can introduce a left-twisted bimodule ${}^\vartheta \!A$. 

We are interested in constructing deformations of the bimodule $A^\vartheta$ in the category of minimal $A_\infty$-algebras. As explained in Sec. \ref{co-res}, this can be done by means of a suitable resolution of the associated graded algebra $\mathcal{A}=\mathcal{A}^0\oplus\mathcal{A}^1$, where $\mathcal{A}^0=A$ and $\mathcal{A}^1=A^\vartheta$. Quite apparently, the choice of a resolution is highly ambiguous  and different resolutions may generate different classes of deformations. Below, we consider only two simple constructions. 

\subsection{Polynomial and Weyl bimodules} \label{PWB} Given a symmetric algebra   $A=S(V^\ast)$, we introduce the algebra of endomorphisms $\mathrm{Hom}(A,A)$ of the $k$-vector space $A$, the product being the composition of endomorphisms. Letting $\Lambda(V)$ denote the exterior algebra of $V$, we define the algebra ${B}=\mathrm{Hom}(A,A)\otimes \Lambda(V)$. The standard grading on $\Lambda(V)$ makes ${B}$ into a graded associative algebra with  ${B}^l=\mathrm{Hom}(A,A)\otimes \Lambda^l(V)$.  

Choosing linear coordinates $\{x^i\}$ on $V$ and $\{p_i\}$ on $V^\ast$, we can identify $S(V^\ast)$ with the algebra of polynomials $k[x^1,\ldots, x^n]$.
Then  the $k$-vector space $\mathrm{Hom}(A, A)$ appears to be  isomorphic to the space of formal power series in $p$'s with coefficients in polynomial functions in $x$'s. Upon this identification, the composition of two endomorphisms $a(x,p)$ and $b(x,p)$ is described by the 
Moyal-type product\footnote{One can think of $a(x,p)$ as a normal symbol of a differential operator on $A=k[x^1,\ldots,x^n]$ (perhaps of infinite order). Then the $\bullet$-product corresponds to the composition of differential operators. }
\begin{equation}\label{b-prod}
a\bullet b=a\exp\left({\frac{\stackrel{\leftarrow}{\partial}}{\partial p_i}\frac{\stackrel{\rightarrow}{\partial}}{\partial x^i}}\right)b\,,
\end{equation}
and homogeneous elements of ${B}^l$ are represented by differential forms 
\begin{equation}\label{om}
\omega =\omega^{i_1\cdots i_l}(x,p)dp_{i_1}\wedge\cdots \wedge dp_{i_l}\,.
\end{equation}
By abuse of notation, we use the same symbol $\bullet$ to denote the multiplication  in ${B}$. 

The usual exterior differential $d: B^l\rightarrow B^{l+1}$ with respect to $p$'s,  
\begin{equation}\label{dom}
d\omega =\frac{\partial \omega^{i_1\cdots i_l}}{\partial p_j} dp_j\wedge dp_{i_1}\wedge\cdots \wedge dp_{i_l}\,,
\end{equation}
makes ${B}$ into a differential graded algebra. Using the standard contracting homotopy $h: B^{l}\rightarrow B^{l-1}$, 
\begin{equation}\label{ch}
h\omega=\int_0^1 {dt}t^{l-1}\omega^{i_1\cdots i_l}(x,tp)p_{i_1}dp_{i_2}\wedge\cdots\wedge dp_{i_l}\,,
\end{equation}
one can see that the differential (\ref{dom}) is acyclic in positive degrees and $H^0({B},d)\simeq A$. Thus, the triple $({B}, \bullet, d)$ provides us with a resolution of the associative algebra $A=S(V^\ast)$. To make contact with the notation of the previous sections, one should shift the degree of $B$ by $-1$ in order for the $\bullet$-product to define the map $m'=m'_2$, see (\ref{bbb}).  We will write $\bar{B}=B[1]$ for the desuspension of $B$.

By dimensional reason, all the $A_\infty$-structures on $\bar{B}$ must belong to the subspace $\bigoplus_{k=0,1,2} \mathrm{Hom}(T^k(\bar{B}),\bar{B})$. In particular, the  $A_\infty$-structures of $ \mathrm{Hom}(T^0(\bar{B}),\bar{B})\simeq \bar{B}$, being necessarily of resolution degree $2$, are represented by $2$-forms
\begin{equation}\label{lambda}
    \lambda=\lambda^{ij}(x,p)dp_i\wedge dp_j\,.
\end{equation}
These forms automatically satisfy the defining condition $[\lambda,\lambda]=0$, while compatibility with the $\bullet$-product requires $\lambda$ to lie in the center of the algebra $B$. Since $Z(B)=k\otimes \Lambda(V)\simeq\Lambda(V)$, the compatible $A_\infty$-structures of $\mathrm{Hom}(T^0(\bar{B}),\bar{B})$ are given by $2$-forms (\ref{lambda}) with constant coefficients. Hence, $d\lambda=0$. Using the contracting homotopy (\ref{ch}), one can readily see that the first-order deformation (\ref{dm}) of $A$ is given by the Poisson bracket 
$$
\mu^{_{(1)}}(a,b)=\frac12 \lambda^{ij}\frac{\partial a}{\partial x^i}\frac{\partial b}{\partial x^j}\,,\qquad \forall a,b\in k[x^1,\ldots,x^n]\,.
$$
The whole deformation, being constructed by formulas (\ref{dif-eq}), reproduces the Moyal $\ast$-product
\begin{equation}\label{WM}
a\ast b=a\exp\left(\frac{t}{2}\lambda^{ij}{\frac{\stackrel{\leftarrow}{\partial}}{\partial x^i}\frac{\stackrel{\rightarrow}{\partial}}{\partial x^j}}\right)b\,.
\end{equation}

In order to construct more interesting examples of minimal deformations, e.g. involving  higher structure maps, we should extend the algebra $A$ by its bimodule $A^{\vartheta}$. A suitable resolution of the extended algebra $\mathcal{A}=A\oplus A^{\vartheta}$ is obtained as follows. The action of $\vartheta$ on $V$ induces the action on the dual space $V^\ast$ and then on the space $B$. Notice that the $\bullet$-product on $B$ is $\vartheta$-invariant. This allows us to define the $\vartheta$-twisted bimodule $B^\vartheta$ over $B$ as well as the trivial extension $\mathcal{B}=B\oplus B^\vartheta$, where the first and second summands have degrees $0$ and $1$, respectively. The product of two elements of  $\mathcal{B}$ reads
\begin{equation}\label{W-mod}
(a,b)\bullet (a',b'):=(a\bullet a',a\bullet b'+b\bullet{}^{\vartheta}\!a')\,.
\end{equation}
The action of the differential (\ref{dom}) extends to $\mathcal{B}$ in the following way: 
$$d(a,b)=(da,-db)\,.$$ 
It is obvious that $H(\mathcal{B},d)\simeq \mathcal{A}$. Hence, upon desuspension,  the differential graded algebra $({\mathcal{B}},\bullet, d)$ provides a resolution of its cohomology algebra $\mathcal{A}=A\oplus A^{\vartheta}$.

Note that a constant $2$-form $\lambda=\lambda^{ij}dp_i\wedge dp_j\in B$ belongs to the center of ${\mathcal{B}}$ iff it is $\vartheta$-invariant.
Any such form defines an $A_\infty$-structure, which is compatible with the $\bullet$-product (\ref{W-mod}). Converse  is also true: any compatible $A_\infty$-structure  $\lambda\in B\subset \mathcal{B}$ generating a minimal deformation of $\mathcal{A}$ is given by a $\vartheta$-invariant $2$-form with constant coefficients.  Applying now the general formulas of Lemma \ref{lem2} together with the contracting homotopy (\ref{ch}),  one can easily see that the corresponding deformation of $\mathcal{A}$ is defined by the Moyal $\ast$-product (\ref{WM}). More precisely, 
\begin{equation}\label{WbM}
(a,b)\ast (a',b'):=(a\ast a', a\ast b'+b\ast {}^{\vartheta}\!a')
\end{equation}
for $a,a',b,b'\in A=k[x^1,\ldots, x^n]$. Again, this deformation gives no higher structure maps. 

Consider now minimal deformations that come from $A_\infty$-structures living in the space $\mathrm{Hom}(T^{1}(\bar{\mathcal{B}}), \bar{\mathcal{B}})$. Each such  structure defines and is defined by a differential $D: {\mathcal{B}}\rightarrow {\mathcal{B}}$ that  commutes with $d$. 
Let us examine the differentials of the form 
\begin{equation}\label{Dab}
D(a, b)=(b\bullet \gamma, 0)\,,
\end{equation}
where $(a,b)\in \mathcal{B}$ and $\gamma=\gamma^{ij}(x,p)dp_i\wedge dp_j$ is some  $2$-form of $B$. It is clear that $D^2=0$. 
Verification of the Leibniz identity for $D$ and the $\bullet$-product (\ref{W-mod}) leads to the following conditions on $\gamma$:
\begin{equation}\label{ComC}
 {}^\vartheta\!\gamma=\gamma\,,\qquad \gamma \bullet a={}^{\vartheta}\!a\bullet \gamma\,,\qquad \forall a\in B\,.
\end{equation}
The second condition is enough to check only for the generators $x^i$ and $p_i$. Let us assume that the automorphism $\vartheta: V\rightarrow V$ is diagonalizable, so that 
$$
{}^\vartheta\! p_i=q_i p_i\,,\qquad {}^{\vartheta}\!x^i=q^{-1}_i x^i
$$
for some nonzero $q_i\in k$. The direct check of (\ref{ComC}) for the generators gives the differential equations 
$$
(q^{-1}_i-1)x^i\gamma=\frac{\partial \gamma}{\partial p_i} \,, \qquad (q^{-1}_i-1)p_i\gamma=\frac{\partial \gamma}{\partial x^i}
$$
with the general solution 
$$
\gamma =e^{\sum_{i=1}^n(q^{-1}_i-1)x^ip_i} \lambda\,,
$$
 $\lambda=\lambda^{ij}dp_i\wedge dp_j$ being a $2$-form with constant coefficients. Then the  first condition in (\ref{ComC}) requires the form $\lambda$ to be $\vartheta$-invariant. 
Finally, the requirement $[D,d]=0$ leads to the closedness condition 
$$
d\gamma=0\qquad \Leftrightarrow\qquad {\sum_{i=1}^n(q^{-1}_i-1)x^idp_i}\wedge \lambda=0\,.
$$
To satisfy this last equation we have to assume that only two eigenvalues of $\vartheta$ are different from $1$, say $q_1$ and $q_2$. Then we can take $\lambda=dp_1\wedge dp_2$. It is clear that 
${}^{\vartheta}\lambda=\lambda$ iff $q_1$ and $q_2$ are mutually inverse to each other, so that 
$$
\gamma=e^{(q^{-1}-1)x^1p_1+(q-1)x^2p_2} dp_1\wedge dp_2
$$
for some $q>1$. Upon substitution to (\ref{Dab}), this $\gamma$ generates a nontrivial deformation of the algebra $\mathcal{A}$.  Furthermore,  the first-order deformation $\mu^{_{(1)}}$ gives rise to the third structure map $m_3$. An explicit expression for $m_3$ is obtained by the general formula (\ref{dm}), where $\partial'$ is the Hochschild differential associated to the associative product (\ref{W-mod}), (\ref{bbb}) and $\tilde{h}$ is determined  by (\ref{ch}). After long but straightforward calculations one can find 
\begin{equation}\label{fff}
\begin{array}{rl}
m_3(\alpha_1\otimes \alpha_2\otimes \alpha_3)&=(-1)^{|\alpha_2|}h\big(hD(\alpha_1)\bullet\alpha_2\big)\bullet\alpha_3\\[3mm]
&=\big(b_1\phi(a_2,a_3), \;b_1\phi(a_2,b_3)-b_1 \phi( b_2,{}^{\vartheta}\! a_3)\big) \,.
\end{array}
\end{equation}
Here $\alpha_i=(a_i,b_i)\in \mathcal{A}$ and we introduced the notation 

\begin{equation}\label{fi}
\phi(a,b)=-\varepsilon^{\alpha\beta}\!\!\!\!\!\!\!\!\!\int\limits_{0< u<w< 1}\!\!\!\!\!\!\!\! dudw\Big(\frac{\partial a}{\partial x^\alpha}\Big)\big((1-w)x+w^{\vartheta}\!x\big)\Big(\frac{\partial b}{\partial x^\beta}\Big)\big((1-u)x+u^{\vartheta}\!x\big),
\end{equation}

$$
\alpha,\beta=1,2 \,,\qquad \varepsilon^{\alpha\beta}=-\varepsilon^{\beta\alpha}\,,\qquad \varepsilon^{12}=1\,,\qquad \vartheta=\mathrm{diag}(q,q^{-1},1,\ldots,1)\,.
$$

Thus, whenever $\mathrm{rank}(\vartheta-1)=2$ and $\det\vartheta=1$, there are two families of deformations of the algebra  $\mathcal{A}$: the first one is generated by the central $2$-forms $\lambda\in Z(\mathcal{B})$, while the second is determined by the differentials $D\in \mathrm{Der}(\mathcal{B})$ of the form (\ref{Dab}). Since $\lambda\in B\subset\mathcal{B}$, $D\lambda =0$. This means that both the $A_\infty$-structures on $\mathcal{B}$ are compatible to each other and we may  consider a $2$-parameter family of deformations generated by $t\lambda + s D$. As we have seen, the $\lambda$-deformation just replaces the usual commutative multiplication of polynomials with the Moyal product (\ref{WM}). Actually, the Moyal deformation is not formal as for any given $a,b\in A$ the series (\ref{WM}) contains only finitely many terms. Hence, we can equate  $t$ to any element of $k$, say $2$. 
Suppose further that the form $\lambda$ is non-degenerate and $\omega=\lambda^{-1}$. Then $(V,\omega)$ is a symplectic vector space endowed with a symplectomorphism $\vartheta\in \mathrm{Sp}(V)$. For  $t=2$, $s=0$ the aforementioned  family of deformations  degenerates to a bimodule over the polynomial Weyl algebra, where the right action is twisted by $\vartheta$.   Letting now $s$ to be  a nonzero parameter, we get a formal deformation of the Weyl bimodule (\ref{WbM}) in the category of  $A_\infty$-algebras. One could arrive at this deformation directly starting from a resolution of the Weyl bimodule.  It turns out that an appropriate resolution is obtained from $(\mathcal{B},\bullet,d)$ by a mere  replacement of the $\bullet$-product (\ref{b-prod}) with the following one:
$$
a\bullet b=a\exp\left(
{\frac{\stackrel{\leftarrow}{\partial}}{\partial p_i}\frac{\stackrel{\rightarrow}{\partial}}{\partial x^i}}+  \lambda^{ij}{\frac{\stackrel{\leftarrow}{\partial}}{\partial x^i}\frac{\stackrel{\rightarrow}{\partial}}{\partial x^j}}\right)b\,,
$$
In \cite{ShSk}, this resolution was called the  { Vasiliev resolution}. The differential (\ref{Dab}) is also modified. To satisfy Eq. (\ref{ComC}) we should now take 
$$
\gamma=e^{\langle p,{}^\vartheta\! x-x\rangle + \lambda(p,{}^{\vartheta}\! p)}\lambda_\vartheta(dp,dp)\,,\qquad
\lambda_\vartheta(dp,dp)=\lambda(dp-{}^{\vartheta}\!dp,dp-{}^{\vartheta}\!dp)\,.
$$
Here the triangle brackets denote the natural pairing and $\lambda(u,v)=\lambda^{ij}u_iv_j$. Then the  first-order deformation $\mu^{_{(1)}}=m_3$ has a more complicated form 
$$
\begin{array}{rl}
m_3(\alpha_1\otimes \alpha_2\otimes \alpha_3)&=(-1)^{|\alpha_2|}h\big(hD(\alpha_1)\bullet\alpha_2\big)\bullet\alpha_3\\[3mm]
&=\big(b_1\ast \Phi(a_2,a_3), \;b_1\ast\Phi(a_2,b_3)-b_1\ast\Phi (b_2,{}^{\vartheta}\! a_3)\big) \,,
\end{array}
$$
where 
\begin{equation}\label{Fi}
\Phi(a,b)=-\int\limits_{0<u<w<1}dudw e^{\langle p_1,(1-w)x+w{}^{\vartheta}\!x\rangle+\langle p_2,(1-u)x+u{}^{\vartheta}\!x\rangle}
\end{equation}
$$
\times e^{\lambda(p_1,p_2)+w\lambda(p_1, {}^{\vartheta}\!p_1+{}^{\vartheta}\!p_2-p_2)+u\lambda(p_2, {}^{\vartheta}\!p_2+{}^{\vartheta}\!p_1-p_1)}
$$
$$
\left.\times e^{w^2\lambda(p_1,{}^{\vartheta}\!p_1) +uw[\lambda(p_1,{}^{\vartheta}\!p_2)+\lambda(p_2,{}^{\vartheta}\!p_1)]+u^2\lambda(p_2,{}^{\vartheta}\!p_2)}\lambda_{\vartheta}(p_1,p_2)a(x_1) b(x_2)\right|_{x_1=x_2=0}\,,
$$
$$
p_1=\left\{\frac{\partial}{\partial x_1^i}\right\}\,,\qquad p_2=\left\{\frac{\partial}{\partial x_2^i}\right\}\,.
$$
By construction, $m_3$ is a nontrivial Hochschild cocycle representing an element of  $HH^3(\mathcal{A}, \mathcal{A})$. 
Writing down the closedness condition for $m_3$, one can see that it is equivalent to the fact that $\Phi$ is a $2$-cocycle of the Weyl algebra $A$  with values in the left-twisted bimodule ${}^\vartheta \!A$, i.e.,
$$
{}^{\vartheta}\!a\ast\Phi(b,c)-\Phi(a\ast b,c)+\Phi(a,b\ast c)-\Phi(a,b)\ast c=0\,.
$$
Such cocycles are closely related to the symplectic reflection algebras \cite{EG}; their integral representation (\ref{Fi})  was first derived in  \cite{SkSh}.  Switching off the Moyal deformation by setting $\lambda=0$ in the exponential functions (\ref{Fi}),  we come back to the expression (\ref{fi}).

\subsection{Quantum polynomial superalgebras} In this section, we will generalize the above example of deformation in two directions. For one thing, we will consider more general extensions of polynomial algebras that involve several automorphisms; for another, we will introduce more general class of resolutions to deform these algebras.  

Let $V$ be an $n$-dimensional vector space over $k$ and let $\Gamma \subset GL(V)$ be a finitely generated, abelian subgroup acting semi-simply on $V$. The group $\Gamma$, being finitely generated and abelian, is isomorphic to the direct product $\mathbb{Z}^k\times \mathbb{Z}_{k_1}\times \mathbb{Z}_{k_2}\times\cdots\times \mathbb{Z}_{k_l}$. Let $\{\vartheta_1,\vartheta_2,\ldots,\vartheta_m\}\subset \Gamma$ denote the generators of $\Gamma$. Since the action of $\Gamma$ in $V$ is semi-simple,  one can chose a basis $\{p_i\}\subset V$ in such a way that 
$$
\vartheta_a p_i=q_{ai}p_i\,,\qquad  i=1,\ldots, n\,,\quad a=1,\ldots, m\,,
$$
for some nonzero  $q_{ai}\in k$. The action of $\Gamma$ extends naturally to the symmetric algebra $S(V)\simeq k[p_1,\ldots,p_n]$.  Geometrically, one can regard the generators $p_i$ as coordinates on the dual vector space $V^\ast$.  

Given the group $\Gamma$, we extend the vector space $V^\ast$ to a quantum superspace $W$ by adding $m$ `odd coordinates' $\pi_a$. The coordinates are assumed to satisfy the commutation relations
\begin{equation}\label{xx}
    p_ip_j-p_jp_i=0\,,\quad \pi_a p_i-q_{ai} p_i\pi_a=0\,,\quad \pi_a\pi_b-\pi_b\pi_a=0\,,\quad (\pi_a)^2=0\,,
\end{equation}
and we prescribe them the following degrees:
$$
|p_i|=0\,,\qquad |\pi_a|=-1\,.
$$
The Grassmann parity of the coordinates is induced by this $\mathbb{Z}$-gadding. The algebra generated by $p$'s and $\pi$'s satisfies the PBW property, so that any its element can be written as a $p\pi$-ordered polynomial  $f(p,\pi)$. We will refer to this algebra as the {\it algebra of quantum polynomials}  \cite{A}, \cite{R}. 

The quantum superspace $W$ can be endowed with  a differential calculus \cite{D}, \cite{KU}. By definition,  the DG-algebra of differential forms $\Omega(W)=\bigoplus_{p\geq 0} \Omega^p(W)$ is generated by the coordinates $p_i$, $\pi_a$, and their differentials $dp_i$, $d\pi_a$ of degrees 
$$
|dp_i|=1\,,\qquad |d\pi_a|=0\,.
$$
The exterior differential $d:\Omega^p(W)\rightarrow \Omega^{p+1}(W)$ is defined now as a degree 1 derivation of $\Omega(W)$ squaring to zero:\footnote{To simplify formulas, we do not write the wedge product. }
\begin{equation}\label{differ}
d(\alpha\beta)=(d\alpha)\beta+(-1)^{|\alpha|}\alpha d\beta\,,\qquad d^2=0\,.
\end{equation}

The ideal generated by (\ref{xx}) in the free algebra on the generators $p$'s and $\pi$'s should now be extended to an ideal in the differential algebra freely generated by the coordinates and their differentials. 
Applying $d$ to Rels. (\ref{xx}) and assuming the differentials $dp_i$ and $d\pi_a$ to be  linearly independent over $\Omega^0(W)$, we get
\begin{equation}\label{xdx}
\begin{array}{c}
    p_idp_j-dp_j p_i=0\,,\quad \pi_a dp_i  + q_{ai} dp_i\pi_a=0\,,\quad d\pi_a p_i -q_{ai} p_i d\pi_a =0\,, \\[3mm]
\pi_a d\pi_a-d\pi_a \pi_a=0 \,,\qquad \pi_a d\pi_b+d\pi_b \pi_a=0\,,\quad a\neq b\,.
    \end{array}
\end{equation}
From the equation $d^2=0$ it then follows immediately that
\begin{equation}\label{dxdx}
\begin{array}{c}
    dp_idp_j+dp_jdp_i=0\,,\qquad d\pi_a dp_i - q_{ai}dp_i d\pi_a =0\,, \\[3mm]
d\pi_a d\pi_b+d\pi_b d\pi_a=0\,,\quad a\neq b\,.
    \end{array}
\end{equation}
Taken together Rels. (\ref{xx} -- \ref{dxdx})  define the Wess--Zumino (WZ) complex associated to a quantum $R$-matrix obeying the additional condition $R^2=1$, see \cite{WZ}, \cite{D}.  

It is known \cite{D1} that the cohomology of the WZ complex $(\Omega^\bullet(W),d)$ is nested in degree zero, 
\begin{equation}\label{HWZ}
    H^\bullet(\Omega, d)\simeq H^0(\Omega, d)\simeq k\,.
\end{equation}
Moreover,  it is not hard to write a contracting homotopy $h:\Omega^{p}(W)\rightarrow \Omega^{p-1}(W)$ leading to this conclusion, see \cite{D}.

The above WZ complex can further be extended to the so-called {\it quantum Weyl superalgebra} \cite{D}, \cite{KU}, \cite{G-Z}. This is achieved by introducing the partial derivatives $\partial^i,\partial^a: \Omega (W)\rightarrow\Omega(W)$ through the relation 
\begin{equation}
    d=dp_i\partial^i+d\pi_a \partial^a\,. 
\end{equation}
It follows immediately that 
$$
\partial^i p_j=\delta^i_j\,,\qquad \partial^a \pi_b=\delta^a_b\,,\qquad \partial^i\pi_a=0\,,\qquad \partial^a p_i=0\,.
$$
Denoting $\partial^i=x^i$, $\partial^a=\theta^a$ and setting $$|x^i|=0\,,\qquad |\theta^a|=1\,,$$ we define  $\mathcal{B}'$ to be the DG-algebra generated by the elements 
\begin{equation}\label{gen}
x^i,\, \theta^a,\, p_i, \,\pi_a,\, dp_i,\, d\pi_a 
\end{equation}
subject to Rels. (\ref{xx}), (\ref{xdx}), (\ref{dxdx}), and 
$$
\begin{array}{c}
    x^i p_j-p_jx^i=\delta^i_j\,,\quad x^i\pi_a -q_{ai}\pi_a x^i=0\,,
    \\[3mm]
    $$\theta^a \pi_a +\pi_a\theta^a=1\,,\quad \theta^a \pi_b -\pi_b\theta^a=0\,,\quad p_i\theta^a -q_{ai} \theta^a p_i=0\,,\\[3mm]
x^i dp_j-dp_jx^i=0\,,\quad x^id\pi_a-q_{ai}d\pi_a x^i=0\,,\quad  dp_i\theta^a+q_{ai}\theta^a dp_i=0\,,\\[3mm]
\theta^a d\pi_a-d\pi_a \theta^a=0\,,\quad \theta^a d\pi_b+d\pi_b \theta^a=0\,,\quad a\neq b\,.\\[3mm]
x^ix^j-x^jx^i=0\,,\quad \theta^a x^i-q_{ai}x^i\theta^a=0\,,\quad \theta^a\theta^b-\theta^b\theta^a=0\,,\quad a\neq b\,,\quad (\theta^a)^2=0
\end{array}
$$
(no summation over repeated indices).  The action of the differential $d$ extends from $\Omega(W)$ to $\mathcal{B}'$ by setting $dx^i=d\theta^a=0$. The DG-algebra $(\mathcal{B}',d)$ enjoys the PBW property and we can represent its elements by ordered polynomials in the variables (\ref{gen}).  The subalgebra generated by the elements $(x^i, p_j, \theta^a, \pi_b)$ is called the { quantum Weyl superalgebra} \cite{D}, \cite{GZ};  it contains  the subalgebra  $\mathcal{A}=\ker d$ generated by $x$'s and $\theta$'s. The latter is clearly isomorphic to the algebra of quantum polynomials $\Omega^0(W)$. Furthermore, it follows from  (\ref{HWZ}) that $H(\mathcal{B}',d)\simeq \mathcal{A}$.

In order to make $(\mathcal{B}',d)$ into a resolution of the algebra $\mathcal{A}$ we prescribe the following bi-degrees to its generators:
$$
\deg x^i=(0,0)\,,\quad \deg p_i=(0,0)\,,\quad \deg \theta^a=(1,0)\,,\quad \deg \pi_a=(-1,0)\,,
$$
$$
\deg dp_i=(0,1)\,,\qquad \deg d\pi_a=(-1,1)\,.
$$
Then $|a|$ coincides with the total degree of the element $a\in \mathcal{B}'$. Although the pair $(\mathcal{B}', d)$ meets all the defining conditions of a resolution, it appears to be too small to generate nontrivial deformations of $\mathcal{A}$. For this reason we consider its completion, denoted by $\mathcal{B}$, with respect to the ideal generated by $\{p_i\}$. The elements of $\mathcal{B}$ are formal power series in $p$'s with coefficients being polynomial functions in the other variables. 

In case $m=1$, the quantum polynomial superalgebra $\mathcal{A}=\mathcal{A}^0\oplus \mathcal{A}^1$ is clearly isomorphic to the trivial extension of the polynomial algebra $A=k[x^1,\ldots,x^n]$ by the bimodule $A^\vartheta$, where $\vartheta$ is the automorphism associated to a single generator $\theta$ of degree $1$.  This situation has been already considered in the previous subsection. 

Let us now describe the $A_\infty$-structures from $\mathrm{Hom}(T^0(\bar{\mathcal{B}}),\bar{\mathcal{B}})\simeq \bar{\mathcal{B}}$  that are compatible with the associative product  and the differential $d$ in $\mathcal{B}$.  These are given by the  $d$-cocycles belonging to the center of $\mathcal{B}$. First, we note that any nonzero element of the form $\pi_a f^a$ cannot be  a  $d$-cocycle, while an element $\theta^a g_a$ does not belong to the center $Z(\mathcal{B})$ unless it is zero. So, we can restrict ourselves  to $\theta$- and $\pi$-independent elements of $\mathcal{B}$. These constitute a differential subalgebra spanned by the forms 
$$
f=g(x,p,dp) (d\pi_1)^{n_1}(d\pi_2)^{n_2}\cdots (d\pi_m)^{n_m}\,,\qquad n_a=0,1,\ldots
$$
Verifying the commutativity conditions 
\begin{equation}\label{fp}
x^i f-f x^i=0\,, \qquad p_i f-f p_i=0\,,
\end{equation}
we find 
$$
f=e^{-\sum_{i=1}^n (1-\prod_{a=1}^m q^{n_a}_{ai})x^ip_i}g(dp)(d\pi_1)^{n_1}(d\pi_2)^{n_2}\cdots (d\pi_m)^{n_m}
$$
for some differential form $g=g^{i_1\cdots i_s}dp_{i_1}\cdots dp_{i_s}$ with constant coefficients. Renumbering the coordinates $x^i$, if necessary,  we may assume that 
\begin{equation}\label{prod1}
\prod_{a=1}^m q^{n_a}_{ai}=\left\{\begin{array}{ll}
     \neq 1 \,,&  i=1,2,\ldots, k;\\[2mm]
     = 1\,, & i=k+1,\ldots, n\,,
\end{array}\right.
\end{equation}
where $k$ depends on $n_a$. Then the closedness condition $df=0$ restricts the form of basis cocycles to
\begin{equation}\label{f}
f=e^{-\sum_{i=1}^n(1-\prod_{a=1}^m q^{n_a}_{ai})x^ip_i} dp_1\cdots dp_k dp_{\alpha_1}\cdots dp_{\alpha_l} (d\pi_1)^{n_1}\cdots (d\pi_m)^{n_m} 
\end{equation}
for some $\alpha_j>k$. Finally, the conditions 
\begin{equation}\label{fpi}
\theta^a f-(-1)^{|f|}f\theta^a=0\,,\qquad \pi_a f-(-1)^{|f|}f\pi_a=0
\end{equation}
impose the following set of restrictions on the numbers $n_a$ and $\alpha_j$:
\begin{equation}\label{prod2}
q_{a1}q_{a2}\cdots q_{ak}q_{a \alpha_1}\cdots q_{a \alpha_l}(-1)^{\sum_{b\neq a}n_b}=1\,,\qquad \forall a=1,2,\ldots,m\,.
\end{equation}
Applying the differential to the second equations in (\ref{fp}) and (\ref{fpi}) yields the other commutativity conditions
$$
dp_i f-(-1)^{|f|}f dp_i=0\,,\qquad d\pi_a f-fd\pi_a=0\,.
$$
Finally, note that the cocycles (\ref{f}) are all nontrivial when viewed as elements of the subcomplex $(Z(\mathcal{B}), d)$.
In such a way we arrive at the next statement. 

\begin{theorem}
The cohomology group $H(Z(\mathcal{B}),d)$ is generated by the cocycles (\ref{f}) with parameters obeying (\ref{prod1}) and (\ref{prod2}). 
\end{theorem}

\begin{example}\label{example} Let $\Gamma=\mathbb{Z}_2$ act on $V=\mathbb{R}^2$ by the  reflection ${}^\vartheta \!x^i=-x^i$, $i=1,2$. Then the algebra $\mathcal{A}=\mathcal{A}_{\mathbb{Z}_2}$ is generated by the three elements $ x^1,x^2$, and $\theta$ subject to the relations
\begin{equation}\label{AZ2}
x^ix^j-x^jx^i=0\,,\qquad x^i\theta +\theta x^i=0\,,\qquad \theta^2=0\,.
\end{equation}
Hence, all $q_i=-1$. The algebra $\mathcal{A}=\mathcal{A}^0\oplus\mathcal{A}^1$ is isomorphic to the trivial extension of the polynomial algebra $A=\mathbb{R}[x^1,x^2]$ by the right-twisted bimodule $A^\vartheta$.  

According to the above considerations, the space $H(\mathcal{B},d)$ is spanned by the cocycles
\begin{equation}
     dp_1dp_2 (d\pi)^{2m}\,,\quad 
    e^{-2x^ip_i}dp_1dp_2 (d\pi)^{2m+1}\,,\quad (d\pi)^{2m}\,,\quad m=0,1,\ldots\,,
\end{equation}
which define mutually compatible $A_\infty$-structures on $\mathcal{B}=\mathcal{B}_{\mathbb{Z}_2}$.

As an associative algebra $H(\mathcal{B},d)$ is generated by the four basis cocyles 
$$
 dp_1dp_2\,,\qquad e^{-2x^ip_i}dp_1dp_2 d\pi\,, \qquad (d\pi)^2\,,\qquad 1\,.
$$
Of these cocycles only the first two have total degree $1$ when regarded as elements of the  algebra $ \bar{\mathcal{B}}$. 
The first cocycle generates the usual Moyal's deformation of the polynomial algebra, while the second leads to  higher structure maps. In particular, the first-order deformation associated to the second cocycle gives a non-zero map $m_3$, which is similar in form to that considered in Sec. \ref{PWB}.

The above result can easily be extended to the Klein group $\Gamma=\mathbb{Z}_2\times \mathbb{Z}_2$ acting on $V=\mathbb{R}^4$ by 
\begin{equation}\label{KG}
    {}^{\vartheta_1}\! x^\alpha =-x^\alpha\,,\qquad {}^{\vartheta_1}\! y^\alpha= y^\alpha\,,\qquad {}^{\vartheta_2 }\! x^\alpha =x^\alpha\,,\qquad {}^{\vartheta_2} \!y^\alpha= -y^\alpha\,,
\end{equation}
$(x^1,x^2,y^1, y^2)$ being coordinates on $\mathbb{R}^4$. The corresponding noncommutative superspace $W$ is obtained by adding the pair of coordinates $\theta^1$ and $\theta^2$ in degree $1$. In fact, the algebra of quantum polynomials on $W$ is given by the tensor product $\mathcal{A}=\mathcal{A}_{\mathbb{Z}_2}\otimes \mathcal{A}_{\mathbb{Z}_2}$, where each factor is isomorphic to the algebra (\ref{AZ2}),  and the  same is true for the resolution algebra $\mathcal{B}=\mathcal{B}_{\mathbb{Z}_2}\otimes \mathcal{B}_{\mathbb{Z}_2}$. By the K\"unneth formula the algebra $H(\mathcal{B},d)$ is multiplicatively generated by the cocycles  
\begin{equation}\label{basC}
\begin{array}{c}
 dx^1dx^2\,,\qquad dy^1dy^2\,,\qquad e^{-2x^\alpha p^x_\alpha}dp^x_1dp^x_2 d\pi_1\,,\qquad e^{-2y^\alpha p^y_\alpha}dp^y_1dp^y_2d\pi_2\,,\\[3mm]
d\theta^1 d\theta^1\,,\qquad d\theta^2d\theta^2\,,\qquad 1\,,
\end{array}
\end{equation}
where $p^x_\alpha$ and $p^y_\alpha$ are coordinates dual to $x^\alpha$ and $y^\alpha$.  The basis cocycles of the first line, when regarded as elements of $ \bar{\mathcal{B}}$, carry total degree $1$. Again, the first two cocycles in (\ref{basC}) generate the Moyal deformation of the algebra $\mathcal{A}=\mathcal{A}_{\mathbb{Z}_2}\otimes \mathcal{A}_{\mathbb{Z}_2}$, which is invariant under the Klein automorphisms (\ref{KG}). This deformation is not formal and we can set the corresponding deformation parameters to $1$ just re-scaling the the generators $x^\alpha$ and $y^\alpha$.  The resulting $\ast$-product algebra (\ref{WbM}) is of primary importance  in $4d$  higher spin theory, where it is called the {\it higher spin algebra}. For an introduction, see \cite{Vasiliev:1999ba}, \cite{Didenko:2014dwa}, \cite{Sharapov:2017yde}.  It is the higher spin algebra alone that dictates the spectrum of massless higher spin fields, the form of free field equations and their gauge symmetries.  The third and fourth cocycles in (\ref{basC}) give rise to a $2$-parameter family of  higher structure maps $m_n$, $n\geq 3$, making $\mathcal{A}$ into a genuine  $A_\infty$-algebra. In the presence of the Moyal deformation,  these higher structure maps correspond to the formal interaction vertices of higher spin fields. It is therefore concluded that all the consistent interactions of massless higher spin fields are controlled by deformations of the higher spin algebra in the category of minimal $A_\infty$-algebras. 

The other nontrivial cocycles of $Z(\mathcal{B})$, while not related to the algebra deformation, may also be of some physical interest. Let us point out the following families of differential forms:
$$
\begin{array}{c}
    dx^1dx^2dy^1dy^2(d\pi_1)^{2m_1}(d\pi_2)^{2m_2}\,,\\[3mm]
     e^{-2x^\alpha p^x_\alpha}dx^1dx^2dy^1dy^2(d\pi_1)^{2m_1+1} (d\pi_2)^{2m_2}\,,\\[3mm]
     e^{ -2y^\alpha p^y_\alpha}dx^1dx^2dy^1dy^2(d\pi_1)^{2m_1} (d\pi_2)^{2m_2+1}\,,\\[3mm]
    e^{-2x^\alpha p^x_\alpha -2y^\alpha p^y_\alpha}dx^1dx^2dy^1dy^2(d\pi_1)^{2m_1+1} (d\pi_2)^{2m_2+1}\,.
\end{array}
$$
These are the cocycles of the maximum resolution degree, namely, $4$. By making use of formula (\ref{dm}), one can convert them to the Hochschild cocycles of the higher spin algebra. From field theoretical standpoint, these Hochschild cocycles correspond to $4$-forms on space-time manifold. When coupled to the Moyal deformation, they may be interpreted as gauge invariant contributions to the on-shell Lagrangian of higher spin fields.  A detailed discussion of these and other physical implications is beyond the scope of this paper. We are going to report on them elsewhere.

\end{example}

\subsection*{Acknowledgments} We are grateful to Vasiliy Dolgushev for a useful discussion and correspondence.


\end{document}